\documentclass[10pt]{article}
\usepackage[letterpaper]{geometry}
\usepackage{hicss}
\usepackage{times}
\usepackage[none]{hyphenat}
\usepackage{url}
\usepackage{latexsym}
\usepackage{indentfirst}
\usepackage{graphicx}
\graphicspath{{images/}}

\usepackage[
    backend=biber,
    style=ieee,
    sorting=none,
  ]{biblatex}
\usepackage{epsfig} 
\usepackage{mathptmx} 
\usepackage{times} 
\usepackage{amsmath} 
\usepackage{amsthm}
\usepackage{amssymb}  
\usepackage{multirow}
\usepackage{booktabs}
\usepackage{tikz-cd}
\usepackage{empheq}

\newcommand{\widefbox}[1]{\fbox{\hspace{1em}#1\hspace{1em}}}
\usepackage[caption=false]{subfig}
\usepackage{accents}
\usepackage[ruled,vlined]{algorithm2e}
\usepackage[cal=cm]{mathalfa} 
\usepackage{makecell}

\newtheorem{theorem}{Theorem}[section]

\newtheorem{proposition}[theorem]{Proposition}

\usepackage{hyperref}
\usepackage{xurl}

\newcommand{\mb}{\mathbf}
\newcommand{\mc}{\mathcal}
\newcommand{\R}{\mathbb{R}}

\newcommand{\diag}{\mathrm{diag}}
\newcommand{\gap}{\mathrm{gap}}

\title{Optimal Measurement Selection for Certifiable Voltage Monitoring in Power Distribution Systems}

\author{Peng Zhang \\
 University of Wisconsin-Madison \\
 {\underline{pzhang286@wisc.edu}} \\ \And
 Line Roald \\
 University of Wisconsin-Madison \\
 {\underline{roald@wisc.edu } } \\ \And
 Manish K. Singh \\
 University of Wisconsin-Madison \\
 {\underline{manish.singh@wisc.edu}} }

\date{}

\begin{document}
\maketitle
\begin{abstract}
Maintaining consumer-end voltages within mandated limits is a central task in operating power distribution systems. Rising penetration of distributed generation, electric vehicles, and flexible loads complicates this task, while a scarcity of real-time measurements leaves distribution networks largely unobservable. Operators commonly monitor a few bellwether meters, which indicate operational health but certify nothing about the unmeasured buses. This work asks how to select measurements that are reported in real time and can certify network-wide voltage-limit compliance. We first illustrate the geometry of the measurement values that certify voltage safety at unmeasured nodes. Building on this geometry, we formulate a bilevel optimization that selects measurements to maximize certification capability, introducing a novel safety-violation metric with favorable monotonicity properties. We derive a strong-duality-based single-level reformulation and two scenario-reduction schemes with quantifiable suboptimality. Numerical tests on the SCE 56-bus system corroborate the superior certification capability and computational efficiency of the proposed approach.
\end{abstract}

\subsubsection*{Keywords:}

Bilevel optimization, sensor selection, voltage monitoring, worst-case analysis.

\section{Introduction} 
\label{sec:intro}

A central task in the design and operation of power distribution systems is maintaining consumer-end voltages within mandated limits. Distribution systems were traditionally designed to keep voltages within the desired range under all operating conditions consistent with load-connection limits. Legacy voltage control relied on seasonally switched capacitor banks and voltage regulators acting on local measurements~\cite{Farivar2011VAR,alejandro2016tap}. This design largely obviated the need for advanced real-time monitoring and control. However, rapid demand growth, rising penetration of distributed generation and electric vehicles, and demand response programs have ushered in the era of active distribution networks \cite{Priyank2023regulation}. Constraining these distributed resources to ranges that ensure voltage limits are met would severely underutilize the existing infrastructure. Active monitoring of distribution systems is therefore of growing interest to researchers and industry alike.

Transmission system monitoring is built on estimating the system state from many near-real-time measurements together with power flow (PF) physics. Distribution system state estimation is considerably harder, primarily because a severe lack of measurements renders the network unobservable \cite{Ahmad2018DSSEsurvey}. Smart meters are now widely deployed \cite{Wang2019SMreview}, yet limited communication bandwidth means locally recorded measurements are reported to central servers only infrequently, for example, once every 24 hours~\cite{kekatos2025solving}. Monitoring of unobservable distribution systems broadly follows two approaches. The first estimates the state using pseudo-measurements, Bayesian priors, or learning-based methods~\cite{Ahmad2018DSSEsurvey,Tong2019bayesian,delavarga25learnDSSE}. Because of the fundamental limitation of unobservability, the accuracy of these estimates rests on modeling assumptions that are typically unverifiable in real time. The second approach, more common in practice, selects a few so-called bellwether meters that report in real time to operators for key inferences~\cite{islam2025smart}. These meters are chosen based on load size and importance, location, available controllable resources, and similar factors.

While indicative of the system's operational health, bellwether meters provide no certifiable guarantees about the unmeasured buses. Voltage-limit compliance was traditionally ensured through system physics and bounded load ranges. Therefore, we select measurements for real-time reporting so that this traditional guarantee is restored. Specifically, we ask: how should measurements be selected for real-time reporting so that network-wide voltage-limit compliance can be certifiably verified? Distribution system sensor selection has a rich literature, but its primary motivations have been to ensure observability, improve estimation accuracy, and detect topology or outages~\cite{Bhela2018observability,Strbac05measurement,Kar2020outageidentification}; we refer the reader to \cite{netto2022review} for a detailed review.

Recent measurement-selection works closest to our goal of verifying voltage-limit compliance include~\cite{chen2024real,buason2024data}, and \cite{molzahn2019grid}. Though motivated differently, all three evaluate the extreme voltages achievable across a distribution network given a small set of real-time measurements. Reference~\cite{molzahn2019grid} uses these extremes to assess the suitability of grid-agnostic controls. Reference~\cite{chen2024real} actively selects measurements in real time to track a network's extreme-voltage nodes. Closest to our motivation, reference~\cite{buason2024data} places a small set of measurements and derives conservative voltage thresholds for them, so that compliance with these tightened limits guarantees network-wide voltage safety. We show that the underlying geometry of safety-certifying measurements can render the safety test of~\cite{buason2024data} overconservative. We therefore advocate an optimization-based safety test and design an accompanying measurement-selection approach. Our measurement selection and safety certification formulations are built using the popular LinDistFlow model for technical simplicity. More generally, extreme achievable voltages can be computed using AC PF models, convex relaxations, or conservative linear approximations~\cite{baran2002network, farivar2013equilibrium, molzahn2019survey}; we defer these extensions for future research. Nevertheless, we numerically validate the LinDistFlow-based safety certifications against AC PF-based voltage data.

The contributions of this work are as follows: \emph{c1)} For a given set of measurement locations, we present geometric insights into the set of measurement values for which voltage safety at unmeasured nodes can be certified. \emph{c2)} We develop a bilevel optimization model that selects measurements for real-time reporting. The selection maximizes voltage-safety certification capabilities for unmeasured buses, quantified by a retrospective measure on historical data. To perform well despite the challenging geometry of the certifying measurement set, we propose a novel metric for violation of safety certification with desirable monotonicity properties in the selected measurements. \emph{c3)} To solve the bilevel problem at scale, we derive a strong-duality-based single-level reformulation. We further develop two scenario-reduction approaches, one based on voltage thresholds and one on iterative constraint generation, that ease the computation while guaranteeing quantifiable suboptimality. \emph{c4)} Numerical tests on the SCE 56-bus distribution system demonstrate the superior certification capability of our method and the computational gains from our simplifications. The scenario-reduction approaches solve the challenging measurement-selection task to optimality in our experiments.

\section{Problem setup} 
\label{sec:model}

This section introduces the linearized PF model under consideration and a worst-case analysis-based voltage-monitoring task. The distribution grid is assumed to have an abundance of smart meters. The meter data is reported with a delay of, say, 24 hours. The premise is that the operator can have real-time access to a limited set of measurements that is insufficient to estimate network-wide voltage magnitudes. Yet they are interested in certifiably ensuring that the unmeasured (i.e., unknown in real time) voltages are within the stipulated limits. We next introduce the problem setup, then describe our approach, which uses historical network-wide voltage data to select the real-time measurements that best enable certifiable voltage monitoring.

\subsection{Network modeling}
Consider a balanced three-phase power distribution network with $N+1$ buses represented by its single-phase equivalent and modeled as an undirected tree graph $\mc G = \{\overline{\mc N}, \mc L\}$. The nodes in $\overline{\mc N} = \mc N \cup \{0\}$ denote the buses, where the root node $0$ represents a substation bus, and the remaining nodes are indexed as $\mc N = \{1, \dots, N\}$. The edges in $\mc L$ denote lines with resistance and reactance represented as $r_\ell$ and $x_\ell$, respectively, for $\ell\in\mc L$. For node $i\in\mc N$, the voltage magnitude is denoted by $v_i$, and $(p_i, q_i)$ represent the net active and reactive power injections. Positive power injections $p_i(q_i) >0$ indicate that power generation exceeds the consumption at node $i$. We then stack quantities across all nodes in $\mc N$ to obtain $N$-length vectors $\mb v, \mb p$, and $\mb q$. Assuming a fixed known substation bus voltage $v_0$, the network-wide voltages can be approximated using the linearized DistFlow equations from \cite{baran2002network, farivar2013equilibrium} as
\vspace{-0.5em}\begin{equation}
    \mb v = v_0 \mathbf{1}_N + \mb R \mb p + \mb X \mb q,
    \label{eq:grid}
\end{equation} where $\mathbf{1}_N \in \R^N$ is the all-ones vector and matrices $(\mb R,~\mb X)$ are from~\cite{farivar2013equilibrium}.

\vspace{-0.5em}
\subsection{Worst-case analysis for voltage monitoring}\vspace{-0.5em}
Consider that a distribution system operator intends to assess voltage safety by checking compliance with specified node-wise lower and upper bounds $[\tilde{\mb v}^{\min}, \tilde{\mb v}^{\max}] \subseteq \R^N$. However, due to the limited communication bandwidth, only the voltage measurements from a subset of buses $\mc N^m \subseteq \mc N$ are available in real time, denoted by $\tilde{v}^m_j$ for $j\in\mc N^m$. While we do not consider real-time power injection measurements, a priori limits are assumed on these based on historical measurements, connection limits, or forecast data. Given noiseless measurements, $\tilde{v}^m_j$, a worst-case analysis over possible power injections can be used to evaluate the extreme achievable voltages at the unmeasured nodes. For instance, under the approximate PF model~\eqref{eq:grid}, the minimum achievable voltage at an unmeasured node $i \in \mc N \backslash \mc N^m$ is given by the linear program (LP) \cite{molzahn2019grid, chen2024real}
\begin{subequations} \label{eq:lp0}
\begin{align}
    \min_{\mb p, \mb q} \quad & \mb e_i^\top (\mb R \mb p + \mb X \mb q) + v_0 \notag\\
    \text{s.t.}\quad  & \underline{\mb p} \le \mb p \le \overline{\mb p}, \label{eq:lp0-2} \\
    & \underline{\mb q} \le \mb q \le \overline{\mb q}, \label{eq:lp0-3} \\
    & \mb e_j^\top (\mb R \mb p + \mb X \mb q) + v_0 = \tilde{v}^m_j, \forall j \in \mc N^m, \label{eq:lp0-4}
\end{align}
\end{subequations}
where $\mb e_i\in\R^N$ denotes the $i$-th canonical vector. Constraints \eqref{eq:lp0-2}-\eqref{eq:lp0-3} bound the net active and reactive power injections, and \eqref{eq:lp0-4} enforces the alignment with observed voltage measurements. The highest achievable voltage can likewise be obtained by maximizing the objective in~\eqref{eq:lp0} while retaining the same feasibility set.

For notational brevity, we denote the power injections as $\mb x = [\mb p^\top \mb q^\top]^\top \in \R^{2N}$ and collect the deviation of voltage measurements, $\{\tilde{v}^m_j-v_0\}_{j\in\mc N^m}$ as $\mb v^m\in\R^{|\mc N^m|}$. Problem~\eqref{eq:lp0} can then be rewritten to define a mapping from measurements to the lowest possible voltage at a node $i$ as
\begin{equation}\label{eq:lp1f}
    f_i(\mb v^m) := \min_{\mb x}~~\mb c_i^\top \mb x~~\text{s.t.}~~\underline{\mb x} \le \mb x \le \overline{\mb x},~~\mb C^m \mb x = \mb v^m,
\end{equation}
where the linear coefficients $(\mb c_i, \mb C^m)$ are obtained from the rows of matrices $(\mb R, \mb X)$ such that~\eqref{eq:lp1f} coincides with~\eqref{eq:lp0}. The mapping yielding the highest achievable voltage at node $i$ can similarly be defined as
\begin{equation}\label{eq:lp1g}
    g_i(\mb v^m) := \max_{\mb x}~~\mb c_i^\top \mb x~~\text{s.t.}~~\underline{\mb x} \le \mb x \le \overline{\mb x},~~\mb C^m \mb x = \mb v^m.
\end{equation}

The LP-based mappings~\eqref{eq:lp1f}--\eqref{eq:lp1g} provide an opportunity to certify voltage safety at unmeasured nodes without precisely estimating these voltages. Specifically, given measurements $\mb v^m$, the voltage safety at node $i$ can be certified if $f_i(\mb v^m)\geq\tilde{v}^{\min}_i-v_0$ and $g_i(\mb v^m)\leq\tilde{v}^{\max}_i-v_0$; hereafter we express the voltage deviation limits as $\mb v^{\min}:=\tilde{\mb v}^{\min}-v_0\mathbf{1}_N$ and ${\mb v}^{\max}:=\tilde{\mb v}^{\max}-v_0\mathbf{1}_N$. The above test is sufficient but not necessary. For instance, an observation $f_i(\mb v^m)<v_i^{\min}$ merely indicates the existence of operating points within power injection limits $[\underline{\mb x}, \overline{\mb x}]$ that would result in measurements $\mb v^m$ while causing an undervoltage at node $i$. In such situations, without additional measurements, true undervoltage incidents cannot be distinguished from potential undervoltages. With the above observation, we pursue the following objective.

\noindent \textbf{Goal:} Given a power network, historical operating conditions, and a measurement budget, find the optimal measurement placement that maximizes the capability of network-wide voltage safety certification by solving LPs~\eqref{eq:lp1f}--\eqref{eq:lp1g} in real time.

\section{Safety-certifying voltage measurement} 
\label{sec:formulation}\vspace{-1em}

To achieve the above-identified goal, this section presents preliminary analysis and design choices that lead to an effective formulation of the problem for optimal measurement selection. 

\subsection{Geometry of safety-certifying voltage measurement sets}\label{sec:geometry}\vspace{-1 em}
The recent work from~\cite{buason2024data} pursues sensor placement towards a voltage safety certification objective similar to ours. In their approach, the safety certification of network-wide voltages was obtained by verifying conservative box constraints on real-time measurements $\mb v^m$ rather than solving the LPs~\eqref{eq:lp1f}--\eqref{eq:lp1g}. The metric for sensor placement was thus related to the dimensions of the hypercube defined by the safety-certifying box constraints. Admittedly, verifying satisfaction of node-wise box constraints in real time saves computation and communication costs when compared to solving~\eqref{eq:lp1f}--\eqref{eq:lp1g}. However, the effectiveness of such an approach hinges on how well box constraints can serve as inner approximations of the set of measurements $\mb v^m$ that can certify voltage safety using~\eqref{eq:lp1f}--\eqref{eq:lp1g}. 

To assess the related set geometries, we present a grid-search-based numerical analysis on the SCE 56-bus system~\cite{farivar2012optimal} where the power injections were sampled as described in Section~\ref{sec:simulation}. Two distinct measurement sets $\mc N^m=\{20,31\}$ and $\mc N^m=\{7,13\}$ were considered for certifying voltage safety on bus~$42$. As per the PF model~\eqref{eq:grid}, the power injection scenarios map to the voltage measurements shown as colored points in Fig.~\ref{fig:intro}. Observe that a hypercube of power injections $[\underline{\mb x}, \overline{\mb x}]$ maps to a compact polytope in the measurement space. Masking the knowledge of true power injections, we solved~\eqref{eq:lp1f}--\eqref{eq:lp1g} using the measurements $\mb v^m$ for each instance and classified them in Fig.~\ref{fig:intro} as follows: the blue points denote instances where $\mb v^m$ certifies satisfaction of upper and lower limits at bus 42; green and orange points indicate potential under- and over-voltages, respectively. The safety-certifying region obtained using the approach of~\cite{buason2024data} is shaded gray.

\begin{figure}[t]
  \centering
  \begin{minipage}[b]{\columnwidth}
    \centering
    \subfloat[$\mc N^m= \{20, 31\}$]
      {\includegraphics[width=0.78\textwidth]{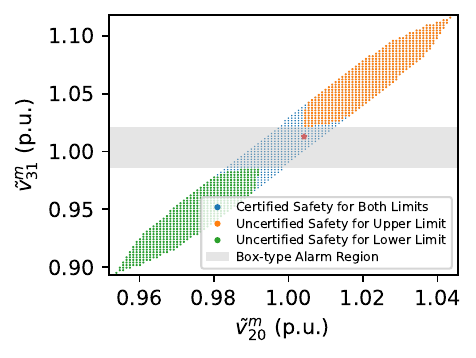}}
  \end{minipage}%
  \vspace{0.2cm}
  \begin{minipage}[b]{\columnwidth}
    \centering
    \subfloat[$\mc N^m= \{7, 13\}$]
      {\includegraphics[width=0.78\textwidth]{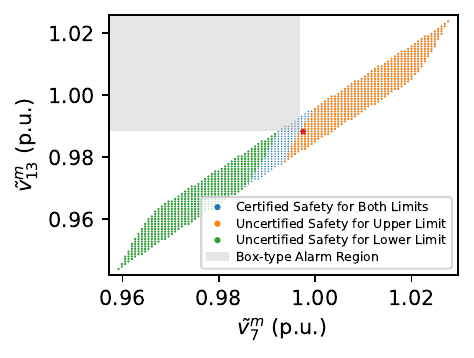}}
  \end{minipage}
  \caption{Classification of voltage measurements from SCE 56-bus system based on certifiability of voltage safety at unmeasured bus $42$.}
  \vspace{-1 em}
  \label{fig:intro}
\end{figure}

The following observations are in order: \emph{i)} The true safety-certifying region can be a non-convex set (see blue dots in Fig.~\ref{fig:intro}~a); \emph{ii)} Many certifiably safe measurements are missed by the box approximations; \emph{iii)} Maximizing the box's linear dimensions may result in real-time safety certifications of limited value, as a significant portion of the measurement space contains physically unrealizable measurements due to bounded power injections (see Fig.~\ref{fig:intro}~b); and \emph{iv)} a particular operating condition, see the red stars in Fig.~\ref{fig:intro} for an example, can lead to safety certification under one set of measurements and could be uninformative under another. Thus, careful selection of the measurement set is well motivated. The above observations demonstrate the limitations in using box constraints for safety certification. Yet, a direct characterization of this non-convex set could be challenging, let alone placing measurements to enlarge it. We therefore propose an alternative metric to guide measurement selection while directly solving~\eqref{eq:lp1f}--\eqref{eq:lp1g} for real-time safety certification.

\subsection{Metric for violation of safety certification}
\vspace{-0.7em}
Recall that given the voltage measurements $\mb v^m$ from a set of measured nodes $\mc N^m \subseteq \mc N$, the minimum and maximum achievable voltages at node $i$ are evaluated as $f_i(\mb v^m)$ and $g_i(\mb v^m)$, respectively. Using these, we quantify the violation of safety certification as \vspace{-1em} \begin{equation}
    h(\mb v^m) = \sum_{i = 1}^N([v_i^{\min} - f_i(\mb v^m)]_+ + [g_i(\mb v^m) - v_i^{\max}]_+).
\end{equation} Note that $f_i(\mb v^m)=g_i(\mb v^m)=v^m_i$ for all measured nodes $i\in\mc N^m$ from~\eqref{eq:lp1f}--\eqref{eq:lp1g}. When network-wide voltage safety can be certified, i.e., all measured voltages and extreme achievable voltages are within stipulated limits, $h(\mb v^m)=0$. Otherwise, the measure $h(\cdot)$ sums up the measured voltage violations and the maximum possible violations at unmeasured nodes.

We will base our measurement selection approach on a retrospective analysis of a historical dataset of network-wide voltage measurements, denoted as $\mc D=\{\mb v_t\}_{t=1}^T$. Given a candidate measurement selection $\mc N^m$, the corresponding measurement dataset can be extracted from $\mc D$ by selecting the respective entries of $\mb v_t$. We denote such dataset as $\mc D^m=\{\mb v^m_t\}_{t=1}^T$. The accumulated violation of safety certification can then be evaluated as
\begin{equation}\label{eq:H1}
   H(\mc N^m;\mc D^m) = \sum_{t = 1}^T h(\mb v^{m}_t),
\end{equation}
and can be used as a measure to assess the suitability of measurement locations $\mc N^m$ for voltage safety certification. To formulate the meter selection task, we introduce a location indicator vector $\mb s\in\{0,1\}^N$ such that $s_i=1,~\forall~i\in\mc N^m$ and $s_i=0$, otherwise. 

We next reformulate \eqref{eq:lp1f}, \eqref{eq:lp1g}, and \eqref{eq:H1} to completely capture the selected measurement-set dependence via vector $\mb s$. To do so, it is helpful to represent the measured voltage information using the big-$M$ formulation as 
\begin{subequations}\label{eq:vs}
\begin{align}
    & \underline{\mb v}_t = \diag(\mb s) \mb v_{t} - M (\mathbf{1}_{N} - \mb s), \label{eq:upp1}\\
    & \overline{\mb v}_{t} = \diag(\mb s) \mb v_{t} + M (\mathbf{1}_{N} - \mb s).  \label{eq:upp2} 
\end{align}
\end{subequations}
The above definition dictates that for the selected measurements with $s_i=1$, $\underline{v}_{t,i}=\overline{v}_{t,i}=v^m_{t,i}$ while for the unmeasured nodes, $\underline{v}_{t,i}=-M$ and $\overline{v}_{t,i}=M$ for some large constant $M$. Note that the dimension of $\mb v^m$ depends on the number of measurements, while the padded vectors $\{\underline{\mb v}, \overline{\mb v}\}$ convey the measurement information without changing their lengths. We can next evaluate the extreme achievable voltages by recasting \eqref{eq:lp1f}--\eqref{eq:lp1g} as 
\begin{subequations}\label{eq:fgtilde}
\begin{align}
    \tilde{f}_i(\underline{\mb v}_{t}, \overline{\mb v}_{t})&= \min_{\mb x}~\mb c_i^\top \mb x~\text{s.t.}~~\mb x\in[\underline{\mb x}, \overline{\mb x}],~\mb C \mb x\in[\underline{\mb v}_{t}, \overline{\mb v}_{t}], \label{eq:low1}\\
    \tilde{g}_i(\underline{\mb v}_{t}, \overline{\mb v}_{t})&= \max_{\mb x}~\mb c_i^\top \mb x~\text{s.t.}~~\mb x\in[\underline{\mb x}, \overline{\mb x}],~\mb C \mb x\in[\underline{\mb v}_{t}, \overline{\mb v}_{t}], \label{eq:low2}
\end{align}
\end{subequations}
where matrix $\mb C=[\mb R~\mb X]$ stems from~\eqref{eq:grid}. Finally, the cumulative measure for violation of safety certification over the historical dataset $\mc D$ becomes a function of $\mb s$ as
\begin{equation}\label{eq:Htilde}
\begin{split}
    \widetilde{H}(\mb s; \mc D) = \sum_{t = 1}^T \sum_{i=1}^N
    &\bigl([v_i^{\min} - \tilde{f}_i(\underline{\mb v}_{t}, \overline{\mb v}_{t})]_+ \\
    &+ [\tilde{g}_i(\underline{\mb v}_{t}, \overline{\mb v}_{t}) - v_i^{\max}]_+\bigr).
\end{split}
\end{equation}
It can be readily verified that if the dataset $\mc D$ does not contain voltage violations and if measurement selection $\mb s$ is adequate to certify voltage safety at all unmeasured nodes for all $T$ scenarios, the cumulative violation metric becomes $\widetilde{H}(\mb s; \mc D)=0$, else $\widetilde{H}(\mb s; \mc D)>0$. The next result shows that metric $\widetilde{H}(\mb s; \mc D)$ features a favorable monotonicity property. Specifically, it is non-increasing with additions in measurement locations. 
\begin{proposition}\label{prop1}
    (Monotonicity) Consider the cumulative metric of violation $\widetilde{H}(\mb s; \mc D)$ in~\eqref{eq:Htilde} and a historical voltage measurement dataset $\mc D$. For a pair of nested measurement sets $\mc N^{m_1} \subseteq \mc N^{m_2}$ corresponding to $\mb s_1\leq\mb s_2$ (component-wise), it holds that $\widetilde{H}(\mb s_1; \mc D)\geq\widetilde{H}(\mb s_2; \mc D)$.
\end{proposition}

\begin{proof}
Consider scenario $t$ from $\mc D$ and evaluate the hypercubes $[\underline{\mb v}_{t,1}, \overline{\mb v}_{t,1}]$ and $[\underline{\mb v}_{t,2}, \overline{\mb v}_{t,2}]$ corresponding to $\mb s_1$ and $\mb s_2$, respectively, using~\eqref{eq:vs}. Assumption $\mb s_1\leq\mb s_2$ entails $[\underline{\mb v}_{t,2}, \overline{\mb v}_{t,2}]\subseteq[\underline{\mb v}_{t,1}, \overline{\mb v}_{t,1}]$. Thus, the computation of extreme achievable voltages using~\eqref{eq:fgtilde} for $\mb s_1$ is a relaxation of the computation for $\mb s_2$, implying
\begin{align*}
    \tilde{f}_i(\underline{\mb v}_{t,1}, \overline{\mb v}_{t,1})&\leq \tilde{f}_i(\underline{\mb v}_{t,2}, \overline{\mb v}_{t,2}),~\forall~i,t\\
    \tilde{g}_i(\underline{\mb v}_{t,1}, \overline{\mb v}_{t,1})&\geq \tilde{g}_i(\underline{\mb v}_{t,2}, \overline{\mb v}_{t,2}),~\forall~i,t.
\end{align*}
Substituting the inequalities in~\eqref{eq:Htilde} establishes the claim~$\widetilde{H}(\mb s_1; \mc D)\geq\widetilde{H}(\mb s_2; \mc D)$.
\end{proof}

The above-identified properties make $\widetilde{H}(\mb s; \mc D)$ a suitable cost function for optimal selection of measurement locations. We next formulate a bilevel optimization-based formulation for measurement selection that minimizes $\widetilde{H}(\mb s; \mc D)$.

\subsection{Bilevel optimization for measurement selection}
Given the historical dataset $\mc D$, and a measurement budget $\rho$, the measurement selection task becomes
    \begin{align}
        \min\quad &\widetilde{H} (\mb s; \mc D)\label{eq:bilevel} \\ 
        \text{over} \quad & {\mb s\in \{0, 1\}^N}, \underline{\mb v}_{t}, \overline{\mb v}_{t} \in \R^{N}, \forall~t,  \notag\\
        \text{s.t.} \quad & \eqref{eq:vs}-\eqref{eq:Htilde},~ \mathbf{1}_{N}^\top \mb s \le \rho. \notag
    \end{align}
The optimization in~\eqref{eq:bilevel} constitutes a bilevel problem where the upper level has a linear objective $\widetilde{H}(\mb s; \mc D)$ with continuous variables $\underline{\mb v}_{t}, \overline{\mb v}_{t}$, binary variables $\mb s$, and mixed-integer constraints. The problem comprises $2NT$ lower-level problems~\eqref{eq:fgtilde} that are LPs. The computational cost of directly solving~\eqref{eq:bilevel} can be limiting even for moderate network and dataset sizes. Therefore, we next develop reformulations and approximations of~\eqref{eq:bilevel} to tractably resolve our measurement selection objective. 

\section{Reformulation and approximations towards computational tractability} 
\label{sec:alg}
Given the bilevel structure, we develop a computationally efficient framework in this section to solve \eqref{eq:bilevel} without sacrificing optimality. By exploiting the special structure on the coupling of upper and lower problems, a single-level optimization is first obtained via the strong duality of \eqref{eq:low1}-\eqref{eq:low2}. To further reduce the computational burden, we develop methods for spatiotemporal subselection of scenarios from $\mc D$, along with mechanisms to assess and ensure the optimality with respect to the full-scale problem.
\subsection{Duality-based single-level reformulation}\vspace{-0.8em}
Consider the following epigraph-based equivalent reformulation of \eqref{eq:bilevel}:
\begin{subequations}\label{eq:single-level}
    \begin{align}
        \min \quad &  \sum_{t = 1}^T \sum_{i =1}^{N}  (l_{i,t} + u_{i,t}) \notag\\
        \text{over}\quad & \mb s, \underline{\mb v}_{t}, \overline{\mb v}_{t}, l_{i,t}, u_{i,t},~\forall~i,~t \notag \\
        \text{s.t.} \quad & \eqref{eq:vs}-\eqref{eq:fgtilde}~\text{and}\notag\\
        &~\mathbf{1}_N^\top \mb s = \rho, \label{eq:rhoeq}\\
        & v_i^{\min} - l_{i,t} \le \tilde{f}_i(\underline{\mb v}_{t}, \overline{\mb v}_{t}),~\forall~i,~t, \label{eq:dual1}\\
        & u_{i,t} + v_i^{\max} \ge\tilde{g}_i(\underline{\mb v}_{t}, \overline{\mb v}_{t}),~\forall~i,~t,  \label{eq:dual2} \\
        & l_{i,t} \ge 0, \: u_{i,t} \ge 0,~\forall~i,~t,
    \end{align}
\end{subequations}
where $(l_{i,t},u_{i,t})$ are non-negative auxiliary variables. Note that constraint~\eqref{eq:rhoeq} replaces the budget inequality in~\eqref{eq:bilevel} with an equality. The replacement is backed by Proposition~\ref{prop1} that guarantees the existence of a minimizer of~\eqref{eq:bilevel} satisfying $\mathbf{1}_N^\top\mb s=\rho$, and simplifies computation by restricting the feasible region. Constraints~\eqref{eq:dual1}-\eqref{eq:dual2} render~\eqref{eq:single-level} a bilevel problem. 

The LPs in~\eqref{eq:fgtilde} are bounded and feasible by construction. Therefore, solving the dual problems yields a zero duality gap. Hence, $\tilde{f}_i(\underline{\mb v}_{t}, \overline{\mb v}_{t})$ can be alternatively computed via the maximization task \vspace{-0.5em}
    \begin{align}\label{eq:dual3}
    \tilde{f}_i(\underline{\mb v}_{t}, \overline{\mb v}_{t})=& \max\quad \underline{\alpha}_{i,t}^\top \underline{\mb x} - \overline{\alpha}_{i,t}^\top  \overline{\mb x} + \underline{\beta}_{i,t}^\top \underline{\mb v}_{t}  - \overline{\beta}_{i,t}^\top \overline{\mb v}_{t} \\
    & \text{over} \quad \underline{\alpha}_{i,t}, \overline{\alpha}_{i,t}, \underline{\beta}_{i,t}, \overline{\beta}_{i,t} \ge 0 \notag\\
    & \text{s.t.} \quad  \mb c_i - \underline{\alpha}_{i,t} + \overline{\alpha}_{i,t} - \mb C^\top( \underline{\beta}_{i,t} - \overline{\beta}_{i,t}) = 0.\notag
    \end{align}   
Similarly, the dual problem for computing $\tilde{g}_i(\underline{\mb v}_{t}, \overline{\mb v}_{t})$ is \vspace{-0.5em}
\begin{align}\label{eq:dual4}
    \tilde{g}_i(\underline{\mb v}_{t}, \overline{\mb v}_{t})=& \min\quad \overline{\gamma}_{i,t}^\top \overline{\mb x} - \underline{\gamma}_{i,t}^\top  \underline{\mb x} + \overline{\delta}_{i,t}^\top \overline{\mb v}_{t} -\underline{\delta}_{i,t}^\top \underline{\mb v}_{t} \\
    & \text{over} \quad \underline{\gamma}_{i,t}, \overline{\gamma}_{i,t}, \underline{\delta}_{i,t}, \overline{\delta}_{i,t} \ge 0 \notag\\
    & \text{s.t.} \quad  \mb c_i + \underline{\gamma}_{i,t} - \overline{\gamma}_{i,t} + \mb C^\top( \underline{\delta}_{i,t} - \overline{\delta}_{i,t}) = 0.\notag
    \end{align}   
If one replaces constraint~\eqref{eq:low1} for computing $\tilde{f}_i(\underline{\mb v}_{t}, \overline{\mb v}_{t})$ in problem~\eqref{eq:single-level} by the dual formulation~\eqref{eq:dual3}, the following advantage surfaces: since constraint~\eqref{eq:dual1} enforces the maximum value from~\eqref{eq:dual3} to be larger than the left hand side, the maximization task can be replaced by the feasibility conditions of~\eqref{eq:dual3}. Similar arguments hold for evaluating $\tilde{g}_i(\underline{\mb v}_{t}, \overline{\mb v}_{t})$ in~\eqref{eq:dual2}. Thus, constraints~\eqref{eq:dual1}--\eqref{eq:dual2} are equivalent to
\begin{subequations}\label{eq:dual56}
\begin{align}
    & v_i^{\min} - l_{i,t} \le \underline{\alpha}_{i,t}^\top \underline{\mb x} - \overline{\alpha}_{i,t}^\top  \overline{\mb x} + \underline{\beta}_{i,t}^\top \underline{\mb v}_{t}  - \overline{\beta}_{i,t}^\top \overline{\mb v}_{t}, \label{eq:dual5}\\
    & u_{i,t} + v_i^{\max} \ge \overline{\gamma}_{i,t}^\top \overline{\mb x} - \underline{\gamma}_{i,t}^\top  \underline{\mb x} + \overline{\delta}_{i,t}^\top \overline{\mb v}_{t} -\underline{\delta}_{i,t}^\top \underline{\mb v}_{t},  \label{eq:dual6} 
\end{align}
\end{subequations}
alongside the constraints of~\eqref{eq:dual3}--\eqref{eq:dual4}. With the aforementioned substitution, problem~\eqref{eq:single-level} simplifies to a single-level optimization, albeit with bilinear terms in~\eqref{eq:dual56}. Notably, the single-level problem has $\mc O(NT)$ constraints featuring $\mc O(N)$ bilinear terms each. If the bilevel problem were simplified by including the KKT conditions of the lower-level problems, rather than the presented approach that exploits strong duality, the complexity would be higher~\parencite{buason2024data}. Since the computational challenges are also exacerbated by the binary variables $\mb s$, we next pursue scenario-reduction approaches that further simplify the measurement selection task.

\subsection{Problem relaxation via scenario reduction}
We observe that if the minimizer of~\eqref{eq:single-level} adequately certifies voltage safety at bus $i$ for a scenario $t$, then the corresponding constraints in~\eqref{eq:dual1}-\eqref{eq:dual2} (or equivalently,  in \eqref{eq:dual56}) shall remain inactive. In practice, a huge majority of these $2NT$ constraints are inactive at optimality. For computational ease, one could solve the measurement selection problem with fewer constraints, in the hope of still recovering a good solution. Let the retained subset of constraints in~\eqref{eq:dual5} and \eqref{eq:dual6} be indexed by the node-scenario tuples $(i,t)\in\mc R^l$ and $(i,t)\in\mc R^u$, respectively. The simplified single-level problem can be formulated as

\begin{subequations}
\begin{align}
     \min \quad &\sum_{(i,t)\in\mc R^l} l_{i,t} + \sum_{(i,t)\in\mc R^u} u_{i,t}
        \tag{\theparentequation}\label{eq:relax}\\
    \text{s.t.} \quad & \eqref{eq:vs},~\mathbf{1}_N^\top \mb s = \rho, \notag
\end{align}
\begin{empheq}[box=\widefbox]{align}
    & \forall\,(i,t)\in\mc R^l\text{:} \notag\\
    & v_i^{\min} - l_{i,t} \le \underline{\alpha}_{i,t}^\top \underline{\mb x}
        - \overline{\alpha}_{i,t}^\top \overline{\mb x} \notag\\
    &\qquad + \underline{\beta}_{i,t}^\top \underline{\mb v}_{t}
        - \overline{\beta}_{i,t}^\top \overline{\mb v}_{t}, \notag\\
    & \mb c_i - \underline{\alpha}_{i,t} + \overline{\alpha}_{i,t}
        - \mb C^\top(\underline{\beta}_{i,t} - \overline{\beta}_{i,t}) = 0, \notag\\
    & \underline{\alpha}_{i,t},\overline{\alpha}_{i,t},
        \underline{\beta}_{i,t},\overline{\beta}_{i,t}, l_{i,t} \ge 0 \notag
\end{empheq}
\begin{empheq}[box=\widefbox]{align}
    & \forall\,(i,t)\in\mc R^u\text{:} \notag\\
    & u_{i,t} + v_i^{\max} \ge \overline{\gamma}_{i,t}^\top \overline{\mb x}
        - \underline{\gamma}_{i,t}^\top \underline{\mb x} \notag\\
    &\qquad + \overline{\delta}_{i,t}^\top \overline{\mb v}_{t}
        - \underline{\delta}_{i,t}^\top \underline{\mb v}_{t}, \notag\\
    & \mb c_i + \underline{\gamma}_{i,t} - \overline{\gamma}_{i,t}
        + \mb C^\top(\underline{\delta}_{i,t} - \overline{\delta}_{i,t}) = 0, \notag\\
    & \underline{\gamma}_{i,t},\overline{\gamma}_{i,t},
        \underline{\delta}_{i,t},\overline{\delta}_{i,t}, u_{i,t} \ge 0 \notag
\end{empheq}
\end{subequations}
Unless the sets $\mc R^l$ and $\mc R^u$ contain all node-scenario pairs, problem~\eqref{eq:relax} is a relaxation of the original measurement selection problem~\eqref{eq:bilevel}. Let $\mb s_R$ be a minimizer of~\eqref{eq:relax} and $\widetilde{H}_R$ represent the optimum value. It can be verified that all binary vectors satisfying $\mathbf{1}_N^\top\mb s\leq \rho$ are feasible for~\eqref{eq:bilevel} and the corresponding violation metric over the entire dataset $\mc D$ can be evaluated from~\eqref{eq:vs}--\eqref{eq:Htilde} by solving the $2NT$ LPs in~\eqref{eq:fgtilde}. Thus, $\tilde{H}_{R}$ serves as a lower bound for the optimum value of~\eqref{eq:bilevel} while $\widetilde{H}(\mb s_R; \mc D)$ provides an upper bound. The optimality gap is thus upper-bounded by
\begin{equation}\label{eq:gap}
    \gap=\widetilde{H}(\mb s_R; \mc D)-\widetilde{H}_R.
\end{equation}
Obtaining zero $\gap$ certifies optimality of the incumbent solution~$\mb s_R$. Hence, we will evaluate and design our scenario reduction approach using the above gap metric.

\subsection{Scenario reduction methods}\label{sec:scenarioreduction}
This section delineates two approaches for constructing the reduced scenario sets~$(\mc R^l,\mc R^u)$. The general idea is to identify tuples $(i,t)$ from the overall dataset~$\mc D$ that are likely to contribute towards active constraints in~\eqref{eq:dual1}--\eqref{eq:dual2}. The $(i,t)$ instances with actual voltage violations are therefore guaranteed choices. Intuitively, $(i,t)$ instances with voltages close to the thresholds may be considered at higher risk of not being certified safe unless they are selected for direct observation. Hence, our first approach heuristically selects voltage thresholds to screen~$\mc D$ and select $(i,t)$ instances for $(\mc R^l,\mc R^u)$. However, the solution's optimality can be guaranteed only if the voltage safety for all scenarios not included in~$(\mc R^l,\mc R^u)$ can be certified by the incumbent minimizer of~\eqref{eq:relax}. Therefore, our second approach builds on constraint-generation methods for iteratively augmenting~$(\mc R^l,\mc R^u)$ unless optimality can be guaranteed. 

\subsubsection{Fixed voltage threshold-based scenario selection.} 

All instances in $\mc D$ are screened to subselect the reduced scenario sets as 
\begin{equation*}
    \begin{aligned}
        \mc R^l &= \{(i, t) \: | \: v_{t,i} \leq v_i^{\min}+{\sigma},~\forall~i\in~[1,N], t\in~[1,T]\}, \\
        \mc R^u &= \{(i, t) \: | \: v_{t,i} \geq v_i^{\max}-{\sigma},~\forall~i\in~[1,N], t\in~[1,T]\},
        \end{aligned}
\end{equation*}
where $\sigma$ is a nonnegative constant that tunes the size of the reduced scenario sets. It could be made distinct for upper/lower limits and varied across nodes, if needed.

\begin{algorithm}[t]
\caption{Iterative Constraint Generation}\label{alg}
\KwIn{Historical dataset $\mathcal D=\{\mathbf v_t\}_{t=1}^T$;
margin step sizes $(\delta^l,\delta^u)$; optimality-gap tolerance $\epsilon$}
\KwOut{Certified solution $\mathbf s^\star$}
\textbf{Initialization} (Scenarios with violations):
\[
\begin{aligned}
\mathcal R^l_0&=\{(i,t)\mid v_{t,i}\leq v^{\min}_i\},\\
\mathcal R^u_0&=\{(i,t)\mid v_{t,i}\geq v^{\max}_i\}.
\end{aligned}
\]
\For{$k=0,1,\ldots$}{
  \tcp{Reduced problem}
  Solve~\eqref{eq:relax} over $(\mathcal R^l_k,\mathcal R^u_k)$ for minimizer
  $\mathbf s_R^k$ and optimum value $\widetilde H_R^k$\;
  \tcp{Candidate generation (increasing margin $(k+1)\delta$)}
  $\begin{aligned}
  \mathcal C_k^l&=\{(i,t)\notin\mathcal R_k^l \mid v_{t,i}\leq v^{\min}_i+(k+1)\delta^l\},\\
  \mathcal C_k^u&=\{(i,t)\notin\mathcal R_k^u \mid v_{t,i}\geq v^{\max}_i-(k+1)\delta^u\}.
  \end{aligned}$\;
  \tcp{Safety certification of candidates under $\mathbf s_R^k$}
  For each $(i,t)\in\mathcal C_k^l$ solve~\eqref{eq:low1} with~\eqref{eq:vs},
  and for each $(i,t)\in\mathcal C_k^u$ solve~\eqref{eq:low2} with~\eqref{eq:vs}\;
  $\begin{aligned}
  \Delta_k^l&=\{(i,t)\in\mathcal C_k^l \mid \tilde f_i(\underline{\mathbf v}_t,\overline{\mathbf v}_t)\le v^{\min}_i\},\\
  \Delta_k^u&=\{(i,t)\in\mathcal C_k^u \mid \tilde g_i(\underline{\mathbf v}_t,\overline{\mathbf v}_t)\ge v^{\max}_i\}.
  \end{aligned}$\;
  Update reduced scenario set
  $\mathcal R_{k+1}^l=\mathcal R_k^l\cup\Delta_k^l,\quad
   \mathcal R_{k+1}^u=\mathcal R_k^u\cup\Delta_k^u$\;
  \If{$\Delta_k^l=\emptyset$ \textbf{and} $\Delta_k^u=\emptyset$}{
 \tcp{Optimality check} 
    \If{$\widetilde H(\mathbf s_R^k;\mathcal D)-\widetilde H_R^k\le\epsilon$}{
    \Return $\mathbf s_R^k$\;}}
}
\end{algorithm}

\subsubsection{Iterative constraint generation-based scenario selection.} 

The performance of the above fixed voltage threshold-based approach depends critically on the choice of the margin~$\sigma$. High margins would result in large cardinalities of~$(\mc R^l,\mc R^u)$ making~\eqref{eq:relax} computationally expensive, while small margins may lead to an unacceptable optimality gap. To alleviate the challenge of explicitly tuning~$(\sigma)$, we propose Algorithm~\ref{alg} that enlarges the lower and upper voltage margins in steps of $\delta^l$ and $\delta^u$, respectively. At each step, the additional $(i,t)$ scenarios identified with loosened thresholds are checked for voltage safety certification using~\eqref{eq:vs}--\eqref{eq:fgtilde} with the last computed measurement selection via~\eqref{eq:relax}. The instances with uncertified safety are included in $(\mc R^l,\mc R^u)$ and~\eqref{eq:relax} is resolved. The algorithm terminates when the optimality gap is bounded under the predetermined tolerance~$\epsilon$.


\section{Numerical tests} 
\label{sec:simulation}

The developed measurement selection approaches were evaluated using the SCE 56-bus distribution network with parameters from~\cite{farivar2012optimal}. The tests were designed to first evaluate the optimality and scalability of the proposed approach and then compare the obtained measurement selection to two alternative approaches.

\noindent
\textbf{Simulation setup.} The benchmark network was modified by adding seven PV generators at arbitrarily chosen locations marked in Fig.~\ref{fig:grid}. The nominal generation for these installations was chosen in the range~$[600,900]$ kW with power factor~$0.95$. The nominal demands were set at $120\%$ of those reported in~\cite{farivar2012optimal}. 
The dataset of voltages~$\mc D$ for solving~\eqref{eq:relax} was generated by random scaling of the nominal values as described below. Instances of active power loads and generations were created by scaling the respective nominal values for all buses and times independently following uniform distributions~$\mc U[0.4,1.6]$ and~$\mc U[0,2]$, respectively. Reactive power demands were sampled by scaling the nominal values first by scaling factors of corresponding active power to model changes in load levels, and then by an additional factor drawn from a Gaussian distribution~$\mc N(1,0.04)$ to model power-factor variations, which were also saturated to be within $40-160~\%$ range of nominal values. Finally, half of the reactive power generation instances were drawn in a manner similar to the reactive demands. The remaining half were generated by scaling the nominal reactive power generation by factors drawn from~$\mc U[-0.5,2]$ to include instances with reactive power withdrawal from inverters. For boundedness, reactive generation instances were capped to $-50$ to $200\%$ of nominal values. The voltage limits for all nodes were set to $\pm 5\%$ of the nominal values. For the generated samples, $0.05\%$ of the bus-time instances had voltages outside the limits. All numerical tests were conducted on an Apple M1 processor with 8 GB RAM, where all data and code are available at \url{https://github.com/pengzhang233/optimal-measurement-selection-for-voltage-monitoring} ; Gurobi 13.0.2 was used to solve optimization problems.

\begin{table}[t]
\centering
\caption{Complexity-optimality trade-off of fixed threshold-based scenario reduction}
\setlength{\tabcolsep}{3pt}
\begin{tabular}{ccccc}
\toprule
$\sigma$~[p.u.] & $\tilde{H}_R~[\times 10^{-3}]$ & Gap & $|\mc R^l| + |\mc R^u|$  & Time~[s]\\
\midrule
\makecell{$0$} & $5.8 $ & $10^{-3}$ & 6 & 4.4 \\
\makecell{$0.005$} & $6.1 $& $7 \times 10^{-4}$ & 31 & 22.1 \\
\makecell{$0.01$} & $6.8 $& $<10^{-4}$ & 67 & 46.4 \\
\bottomrule
\end{tabular}
\vspace{-0.5em}
\label{table:sigma}
\end{table}

\subsection{Optimality and scalability assessment}
We first evaluated our ability to solve the measurement selection problem~\eqref{eq:relax} using the two scenario-reduction approaches proposed in Section~\ref{sec:scenarioreduction}. The dataset $\mc D$ is constructed with $T=288$ instances, representative of the number of 5-minute periods in a day. The measurement budget for all tests was set to $\rho=8$. The performance of the approach with fixed voltage threshold is expected to depend critically on the choice of margin $\sigma$. We therefore solved~\eqref{eq:relax} for different margin values reported in Table~\ref{table:sigma}. As anticipated, increasing the margins results in a higher number of scenarios selected in $(\mc R^l, \mc R^u)$, a lower optimality gap, and higher solve times. The evaluation of gap follows~\eqref{eq:gap} by solving $2\times 55 \times 288$ LPs in~\eqref{eq:fgtilde} ($\approx 4$~minutes). Based on our threshold for optimality gap $\epsilon=10^{-4}$, we fix $\sigma=0.01$ for the next set of tests.  

\begin{table}[tbp]
\centering
\caption{Performance of two scenario reduction methods averaged over 10 datasets $\mc D$}
\setlength{\tabcolsep}{3pt}
\begin{tabular}{ccccc}
\toprule
& $\tilde{H}_R~[\times 10^{-2}]$ & Gap & $|\mc R^l| + |\mc R^u|$  & Time~[s]\\
\midrule
\makecell{Fixed \\threshold} & $3.3$ & $<10^{-4}$ & 97 & 111.7 \\
\makecell{Constraint \\generation} & $3.3$& $<10^{-4}$ & 21 & 43.6 \\
\bottomrule
\end{tabular}
\label{table:compare}
\end{table}

\begin{figure}[thb]  
    \centering
	\includegraphics[width=0.8\linewidth]{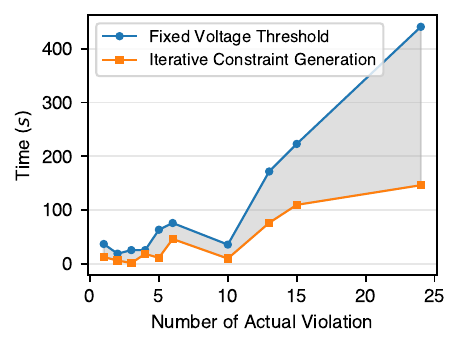}
	\caption{Computational scalability of the proposed scenario reduction methods.}
    \vspace{-1em}
    \label{fig:51}
\end{figure}

To compare the two scenario-reduction approaches, we independently generated 10 datasets $\mc D$, each with $T=288$ instances. For the iterative constraint generation-based approach, we used $\delta^l = \delta^u = 0.005$. Table~\ref{table:compare} reports the cardinalities of the reduced-scenario sets, optimum value and computation time for solving~\eqref{eq:relax}, and the optimality gap~\eqref{eq:gap}, for the two reduction approaches averaged over the 10 datasets. While the computationally taxing optimality check is placed within Algorithm~\ref{alg}, our tests showed that it was run only once at termination, same as the fixed threshold-based approach. To compare the measurement-selection times of the two approaches, we left out the optimality-check time ($\approx 4$~minutes) from the comparison. We observed that the average solve time of~\eqref{eq:relax} with the fixed threshold-based approach is about 2.5 times that with the iterative constraint generation method. Curiously, the average time for the fixed threshold-based approach is much higher than the findings reported in~Table~\ref{table:sigma}. Diving deeper, we observed large variations in the solve times of the fixed threshold-based approach across the ten datasets. Notably, for datasets with a higher number of actual voltage violation instances, indicative of stressed conditions, the fixed threshold-based approach scales poorly; see Fig.~\ref{fig:51}. Overall, the superior optimality and scalability performance of the iterative constraint generation-based approach motivated us to adopt Algorithm~\ref{alg} for the remaining numerical tests.

\subsection{Comparative evaluation of measurement selection approaches}
We conducted a comparative analysis using three frameworks of offline measurement selection and two approaches of real-time safety certification. Concretely, 5 of the 10 datasets $\mc D$ were used for measurement selection ($\mc D_{select}$), while the remaining were used for testing ($\mc D_{test}$), providing distinct datasets for offline computation and online evaluation. Besides Algorithm~\ref{alg}, we implemented a simple baseline ($\textsf{B1}$) that splits the budget $\rho=8$ equally between the nodes witnessing the highest and lowest voltages in $\mc D_{select}$, and a second benchmark ($\textsf{B2}$) based on~\cite{buason2024data} with an identical measurement budget.  For real-time certification, we tested two approaches: ($\textsf{T1}$) solving the LPs~\eqref{eq:lp1f}--\eqref{eq:lp1g}, and ($\textsf{T2}$) checking compliance with the box constraints from ($\textsf{B2}$). In real-time, safety certification via ($\textsf{T2}$) was virtually instantaneous, while ($\textsf{T1}$) took $0.7$~seconds. We first compared $(\textsf{B2})+(\textsf{T2})$ (i.e.~\cite{buason2024data}) against Algorithm~\ref{alg} + ($\textsf{T1}$). Since~($\textsf{T2}$) provides a collective certification for network-wide safety at a given scenario, we deemed the safety uncertified if ($\textsf{T1}$) failed to certify the safety of any node at scenario $t$. The number of uncertified scenarios, out of $T = 288$, for each test dataset is shown in Fig.~\ref{fig52}~\emph{(top)}. We observed that, across all datasets, there were $17$ scenarios with actual voltage violations, which were the only scenarios our approach failed to certify, implying no false alarms. On the other hand, $(\textsf{B2})+(\textsf{T2})$ failed to certify the safety for $58$ scenarios in total.

Guided by the geometric insights in Section~\ref{sec:geometry}, we expected the box constraints of $(\textsf{T2})$ to be the limiting factor in previous tests using the approach from~\cite{buason2024data}. Therefore, we next re-evaluated the selections from Algorithm~\ref{alg}, ($\textsf{B1}$) and ($\textsf{B2}$), all using ($\textsf{T1}$) for real-time safety certification with results reported in Fig.~\ref{fig52}~\emph{(bottom)}, where Fig.~\ref{fig:grid} compares measurement locations. Here ($\textsf{T1}$) certifies safety at nodal granularity for each scenario. Cumulatively over the five test datasets, there were $40$ voltage-violation instances. The number of node-scenario instances left uncertified was $45$ for Algorithm~\ref{alg}, versus $107$ and $57$ for ($\textsf{B1}$) and ($\textsf{B2}$), respectively. Note the slight conservativeness of our approach surfacing at this increased granularity. The poor performance of $(\textsf{B1})$ shows that historically extreme-voltage-experiencing nodes are not necessarily informative of the system-wide voltages. Moreover, although ($\textsf{B2}$) improves under ($\textsf{T1}$) rather than ($\textsf{T2}$), Algorithm~\ref{alg} still outperforms it, confirming the value of informing measurement selection with the real-time test model. We finally validated the performance of LinDistFlow-based ($\textsf{T1}$) against AC PF-based voltage data and observed that all instances certified safe by ($\textsf{T1}$) were indeed within the desired voltage limits.

\begin{figure}[htbp]
  \centering
  \begin{minipage}[b]{\columnwidth}
    \centering
    \subfloat{\includegraphics[width=0.8\textwidth]{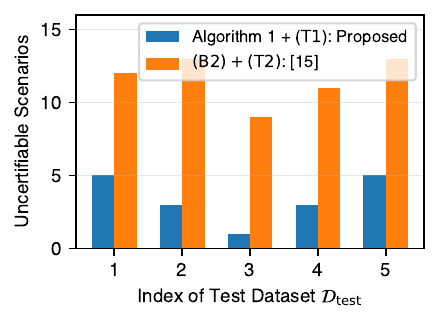}}
  \end{minipage}%
  \vspace{0.cm}
  \begin{minipage}[b]{\columnwidth}
    \centering
    \subfloat{\includegraphics[width=0.8\textwidth]{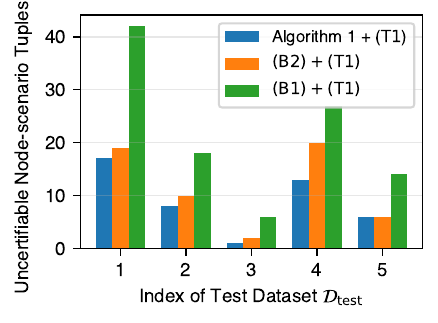}}
  \end{minipage}
  \caption{Number of instances with uncertifiable voltage safety across five test datasets.}
  \vspace{-1 em}
  \label{fig52}
\end{figure}

\begin{figure*}[thb]
    \centering
	\includegraphics[width=0.8\linewidth]{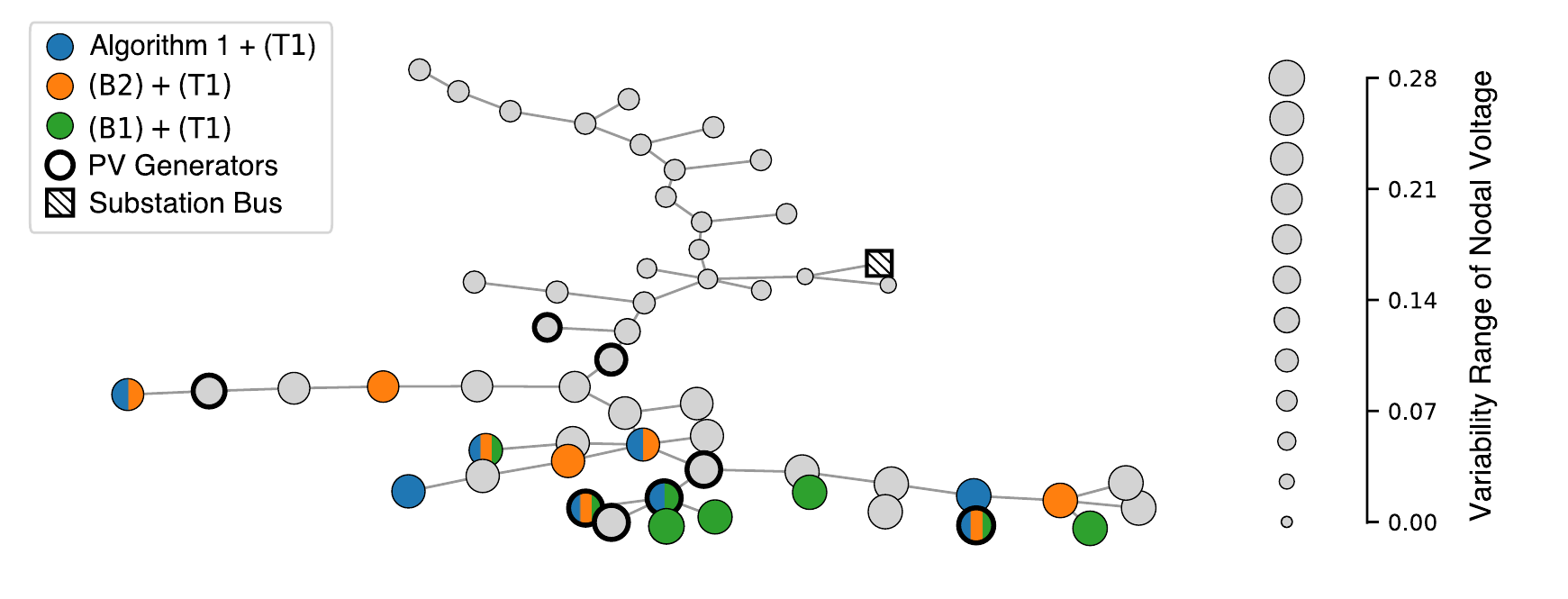}
	\caption{Illustration of selected measurements by different methods on the SCE 56-bus benchmark.}
	\label{fig:grid} 
\end{figure*}

\section{Conclusions} 
\label{sec:conclusion}

This work addressed the selection of a few real-time measurements in unobservable distribution systems that certify network-wide voltage compliance. We characterized the geometry of the safety-certifying measurement set and exposed its non-convexity as a challenge for direct approximation by simple box constraints. This insight motivated a bilevel measurement-selection formulation built on a novel violation metric with a favorable monotonicity property. A strong-duality-based single-level reformulation, together with two scenario-reduction schemes rendered the problem tractable while certifying optimality through a computable gap. Numerical tests on the SCE 56-bus system demonstrated markedly higher certification capability than the competing baselines, with the iterative constraint-generation scheme proving especially scalable. Our future work focuses on extending the framework beyond the LinDistFlow model to higher-accuracy PF representations, accommodating the heterogeneous measurement types available in distribution systems, and modeling measurement noise.



\printbibliography

\end{document}